\documentclass[conference,a4paper]{IEEEtran}

\IEEEoverridecommandlockouts

\newcommand{\bF}{\mathbb{F}}

\usepackage{amsmath,amssymb,amsfonts,amsthm}
\usepackage{multirow}
\usepackage{comment}
\usepackage{graphicx}
\usepackage{textcomp}
\usepackage{xcolor}
\usepackage{algorithm}
\usepackage{algpseudocode}
\usepackage[%
backend=biber,
style=numeric,
citestyle=numeric,
sortcites=true,
natbib=true,
giveninits=true,
maxcitenames=1,
maxbibnames=6,
minbibnames=6,
doi=false,
url=true,
isbn=false
]{biblatex}
\def\BibTeX{{\rm B\kern-.05em{\sc i\kern-.025em b}\kern-.08em
    T\kern-.1667em\lower.7ex\hbox{E}\kern-.125emX}}
\newtheorem{definition}{Definition}
\newtheorem{theorem}{Theorem}
\newtheorem{lemma}{Lemma}
\newtheorem{remark}{Remark}
\newtheorem{example}{Example}

\newif\ifextended
\extendedtrue
\begin{document}
\pagestyle{plain}
\title{
A Generalized Covering\\Algorithm for Chained Codes
}


 \author{\textbf{Ben Langton}$^\star$ and \textbf{Netanel Raviv}$^\dagger$\\
 $^\star$Harvey Mudd College, Claremont, CA 91711, USA\\
 $^\dagger$Washington University in St. Louis, St. Louis, MO 63103, USA.}


\maketitle

\begin{abstract}
The covering radius is a fundamental property of linear codes that characterizes the trade-off between storage and access in linear data-query protocols. The generalized covering radius was recently defined by Elimelech and Schwartz for applications in joint-recovery of linear data-queries. In this work we extend a known bound on the ordinary covering radius to the generalized one for all codes satisfying the chain condition---a known condition which is satisfied by most known families of codes. Given a generator matrix of a special form, we also provide an algorithm which finds codewords which cover the input vector(s) within the distance specified by the bound. For the case of Reed-Muller codes we provide efficient construction of such generator matrices, therefore providing a faster alternative to a previous generalized covering algorithm for Reed-Muller codes.
\end{abstract}


\section{Introduction}
The covering radius of a code is the minimum integer~$r$ such that any vector in the space is within Hamming distance at most~$r$ from a codeword of the code. This fundamental property of codes is very well understood~\cite{Cohen2005CoveringC}, and has applications in low-access algorithms for linear queries in databases \cite{Graham1985OnTC,Cohen2005CoveringC}.

Motivated by joint-recovery of multiple linear queries simultaneously, the \textit{generalized} covering radius was recently introduced in~\cite{elimelech2021generalized,elimelech2022rmgeneralized}. Roughly speaking, the $t$-th generalized covering radius is the maximum number of coordinates in which $t$ vectors can differ from $t$ 
codewords, across all $t$-subsets of the code; for~$t=1$ this definition specifies the (ordinary) covering radius.


While little is known about the generalized covering radius for most codes, upper and lower bounds were established for Reed-Muller codes in \cite{elimelech2022rmgeneralized} and for some other codes and values of $t$ in \cite{elimelech2021generalized}. They also provided an algorithm which, given a set of~$t$ vectors, finds~$t$ codewords within their bound for Reed-Muller codes. 

As noted in \cite{elimelech2021generalized}, the generalized covering radius is closely related to Generalized Hamming Weights (GHW), introduced by Wei in \cite{hammingweights,chaincondition}. In this paper we directly link the two by showing that for every code which satisfies the \textit{chain condition}, the generalized covering radius can be bounded using the GHWs. The chain condition~\cite{chaincondition} asserts that there exists a generator matrix which realizes the GHWs, and is satisfied by most important families of codes.
Our bound follows from combining  the bound in~\cite{coveringradiusbound} for GHWs with one of the equivalent definitions of the generalized covering radius given in~\cite{elimelech2021generalized}.

Our bound implies an efficient algorithm that for a given set of~$t$ vectors, finds a corresponding set of~$t$ codewords which differ from the~$t$ vectors by at most the value of the bound. This algorithm requires a \textit{chained generator matrix}, a generator matrix of a special form which is guaranteed to exist in all chained codes. 

We further show how a chained generator matrix for Reed-Muller codes can be found easily. This results in a generalized covering algorithm for Reed-Muller codes, which exponentially improves the runtime of the one given in~\cite{elimelech2022rmgeneralized}, albeit with a potentially negative impact on performance. Our algorithm also applies to~$q$-ary Reed-Muller codes for~$q>2$, that were not addressed by~\cite{elimelech2022rmgeneralized}.

Finally, since the algorithm provides codewords up to the value of the bound, which might be larger than the generalized covering radius of the code, we computationally compare our bound to the best known ones for Reed-Muller codes~\cite{elimelech2022rmgeneralized}. Our experiments suggest that the given bound outperforms the best known ones for Reed-Muller codes in some parameter regimes.


\section{Preliminaries}

\subsection{The Generalized Covering Radius}

The $t$-th generalized covering radius~$R_t(C)$ \cite{elimelech2021generalized} is defined as follows, where~$[t]$ denotes~$\{1,\ldots,t\}$, and~$\binom{[n]}{r}$ denotes the family of subsets of size~$r$ of~$[n]$.

\begin{definition}\cite{elimelech2021generalized}\label{def:gereralizedradius2}
Let $C$ be an $[n,k]$ code over $\mathbb{F}_q$. For $t \in \mathbb{N}$, the~$t$-th generalized covering radius~$R_t(C)$ is the minimal integer~$r$ such that for every $v_1, v_2, .., v_t \in \mathbb{F}^n_q$, there exist codewords $c_1, c_2, ..., c_t\in C$ and $I \in {[n] \choose{r}}$ such that $\textup{supp}(v_i - c_i) \subseteq I$ for all $i \in [t]$. 
\end{definition} 

It can be readily verified that this definition specifies to the well-known covering radius by setting~$t=1$. Intuitively, the $t$-th generalized covering radius~$r$ is the maximum number of coordinates in which~$t$ vectors can differ from the~$t$ codewords which minimizes this number. The quantity~$R_t(C)$ can be defined in multiple equivalent ways, out of which we make use of the following one in the sequel.

\begin{definition}\cite{elimelech2021generalized}\label{def:generalizedradius4} 
Let $C$ be an $[n,k]$ linear code over $\mathbb{F}_q$ with generator matrix $G$, and $C_t$ be the $[n,k]$ linear code over $\mathbb{F}_{q^t}$ with generator matrix $G$. Then $R_t(C) = R_1(C_t)$.
\end{definition}


We present a few basic results about the generalized covering radius below. We first have that the generalized covering radii are monotone increasing for a given code. We omit the reference to a specific code~$C$ whenever unnecessary.

\begin{theorem}\cite{elimelech2021generalized}\label{monotone}
$R_1 \leq R_2 \leq ... \leq R_{n-k} = n-k$.
\end{theorem}

Clearly, in order to cover any~$t$ given vectors, one can use the ordinary covering radius~$t$ times, which gives rise to the next theorem. The crux of studying~$R_t$ is in cases which this inequality is strict.

\begin{theorem}\cite{elimelech2021generalized}\label{subadditive}
 Let $C$ be an $[n, k]$ code over $\mathbb{F}_q$. Then for all
$t_1$, $t_2 \in \mathbb{N}$, we have $R_{t_1+t_2} \leq R_{t_1} + R_{t_2}$
\end{theorem}

For example, these two theorems readily imply the generalized covering radii of the Hamming code, which is known to have a covering radius of 1. By the fact that $R_1 = 1, R_{n-k} = k$, and Theorems~\ref{monotone} and~\ref{subadditive}, we have that $R_t = t$ for $t \leq n-k$. 


\subsection{Generalized Hamming Weights and Chained Codes}
The generalized Hamming weights, introduced in \cite{elimelech2021generalized}, are a similar extension of minimum distance as the generalized covering radius is to the ordinary covering radius. Recall the definition of a support of a code $\textup{supp}(C) = \{i: \exists (x_1,...,x_n) \in C, x_i \neq 0\}$, and define the following.

\begin{definition} \cite{hammingweights}\label{def:dr} The $r$-th generalized Hamming weight of a code~$C$ is $d_r(C) = \min\{  |\operatorname{supp}(D)|:D\subseteq C, \dim(D)=r\}$.
\end{definition}

The \textit{chain condition} is then defined as follows.

\begin{definition} \label{chaincondition}
\cite{chaincondition}
 An $[n,k]$ linear code $C$ with GHWs $d_1(C),d_2(C),...,d_k(C)$ satisfies the chain condition (abbrv. \emph{chained code}) if there are $k$ linearly independent vectors $c_1,c_2,...c_k$ such that $d_r(C) = |\bigcup_{i=1}^r \textup{supp}(c_i) | $ for every~$r\in\{1,\ldots,k\}$.
\end{definition}

Intuitively, the span of each~$t$ prefix of the basis~$c_1,\ldots,c_k$ is a subcode which realizes the minimum in Definition~\ref{def:dr}. Using such basis as rows of a generator matrix, we define the following.


\begin{definition}\label{matrixchaincondition}
For a chained code~$C$, a generator matrix~$\Gamma$ with rows~$c_1$ (top row) through $c_k$ (bottom row) is called a \emph{chained generator matrix} if each $c_i$ ends with~$n-d_i$ zeros, where $d_i = |\bigcup_{j=1}^i \textup{supp}(c_j)|$. 
\end{definition}

\begin{remark}\label{remark:permuteEntries}
    Given $c_i$'s which realize the GHW hierarchy (i.e., satisfy the condition in Definition~\ref{chaincondition}), we can make each~$c_i$ end with at least $n-d_i$ zeros by permuting the columns so that in each row the $d_i - d_{i-1}$ columns which are new in the support are moved to the end of the nonzero part of each row, starting with $c_1$. Therefore, every chained code has a chained matrix; these matrices will be useful in the sequel for providing a simple generalized covering algorithm.
\end{remark}




For the following theorems, let $C$ be an $[n,k]$ code and $J$ be a subset of the coordinates 
of $C$ with $|J| < n-d_1$. We also assume a generator matrix of $C$ of the form 
\begin{align*}
    \begin{bmatrix}
      g(C_0) & 0 \\ A & g(C_J) 
    \end{bmatrix},
\end{align*}



where $g(C_J)$ is a generator matrix of $C_J$, the projection of $C$ onto the coordinates $J$, and $g(C_0)$ is the generator matrix of $C_0$, the subcode of $C$ which is 0 on $J$ (but does not contain the coordinates in $J$). With this in mind, we give the following lemma, which is used for the generalized covering algorithm. The proof is given for completeness. 


\begin{lemma}\label{additivity}\cite{MATTSON1983453}
$R_1(C) \leq R_1(C_J) + R_1(C_0)$. 
\end{lemma}

\begin{proof}
Let $v=(v_0, v_J)$ be an arbitrary vector of length $n$, where~$v_J$ is of length~$|J|$ and $v_0$ is of length $n - |J|$. Then there exists a codeword of~$C$ of the form $(a, c_J)$, where $c_J \in C_J$ satisfies $d_H(v_J,c_J) \leq R_1(C_j)$ (where~$d_H$ denotes Hamming distance). Furthermore, there is a codeword of the form $(c_0,0)$, where $c_0 \in C_0$, such that $d_H(v_0 + a, c_0) \leq R_1(C_0)$. Therefore,~$v$ is of distance at most $R_1(C_J) + R_1(C_0)$ from $(a,c_J) + ( c_0,0)$.
\end{proof}
Furthermore, as proven in \cite{coveringradiusbound}, the generalized Hamming weights are related to the covering radius by the following bound. A full proof is given in order to clarify subsequent parts of the paper.
\begin{theorem}\label{ghwcoveringbound}
\cite{coveringradiusbound} Let $C$ be an $[n,k]$ chained code with GHWs $d_0=0,d_1,d_2,...,d_k$. Then the covering radius~$R(C)$ of~$C$ satisfies 
$$R(C) \leq n - \sum_{r=1}^k \Bigl\lceil\frac{d_r - d_{r-1}}{q}\Bigr\rceil.$$
\end{theorem} 

\begin{proof} 
Let~$C$ be an $[n,k]$ chained code (Definition~\ref{chaincondition}) with GHWs ${d_1, d_2, \ldots, d_k}$, and let $\Gamma(C)$ be a chained generator matrix of~$C$ with rows~$c_1,\ldots,c_k$ (Definition \ref{matrixchaincondition}), arranged as in Remark~\ref{remark:permuteEntries}. Further, for~$r\in[k]$ let~$M_r$ be the top-left $r\times d_r$ submatrix of~$\Gamma(C)$, let~$s_{r,1},\ldots,s_{r,r}\in\bF_q^{d_r}$ be its rows (numbered top to bottom), and let~$C_r$ be the row-span of~$M_r$. Since~$C_k=C$, the theorem can be proved by induction on the dimension of the code, as follows.


In the base case, notice that~$C_1=\operatorname{span}\{s_{1,1}\}$, and~$s_{1,1}\in\bF_q^{d_1}$ has no zero entries.
Fix an arbitrary vector~$v\in\bF_q^{d_1}$ and denote~$s_{1,1}=(\sigma_1,\ldots,\sigma_{d_1})$. By the pigeonhole principle, the multi-set~$\{\{ v_j/\sigma_{j}\}\}_{j=1}^{d_1}$ contains some element at least $\lceil d_1/q \rceil$ times, and let~$\lambda$ be that element. It follows that~$d_H(\lambda s_{1,1},v)\le d_1-\lceil d_1/q\rceil$, which implies that 
$R(C_1) \leq d_1 - \lceil d_1/q \rceil$ and concludes the base case.

In the inductive step, we assume that $R(C_r) \leq d_r - \sum_{i=1}^r \lceil (d_i - d_{i-1})/q \rceil$ and observe that by construction 
\begin{align*}
    M_{r+1}= 
    \begin{pmatrix}
        \overbrace{M_{r}}^{d_r} & \overbrace{\begin{matrix} 0 & \ldots & 0 \end{matrix}}^{d_{r+1}-d_r} 
         \\ \multicolumn{2}{c}{s_{r+1,r+1}} 
    \end{pmatrix}
\end{align*}


Since the last $d_{r+1} - d_r$ elements of $s_{r+1,r+1}$ must be nonzero (since~$M_{r+1}$ does not have zero columns), we have that 
$$R(C_{r+1}) \leq R(C_r) + d_{r+1} - d_r - \lceil (d_{r+1} - d_{r})/q \rceil$$
by Lemma~\ref{additivity}, and another similar use of the pigeonhole principle. Using our induction hypothesis, this implies that 
\begin{align*}
    R(C_{r+1}) &\leq d_r - \sum_{i=1}^r \lceil (d_i - d_{i-1})/q \rceil \\
    &\phantom{=}+ d_{r+1} - d_r - \lceil (d_{r+1} - d_{r})/q \rceil \\
    &= d_{r+1} - \sum_{i=1}^{r+1} \lceil (d_i - d_{i-1})/q \rceil,
\end{align*}
which completes the proof.
\end{proof} 
A \textit{covering algorithm} is evident from the proof of Theorem~\ref{ghwcoveringbound} and Lemma~\ref{additivity}. The algorithm receives a word~$v\in\bF_q^n$ to cover, and outputs a codeword within distance at most $n - \sum_{r=1}^k \Bigl\lceil\frac{d_r - d_{r-1}}{q}\Bigr\rceil$ from~$v$. The algorithm requires a chained generator matrix, and proceeds by covering~$v$ sequentially by~$C_1,C_2\ldots,C_k$, by finding the proper scalar multiple which covers the corresponding part of~$v$, and subtracting the covering codeword from what is left to cover. 
\ifextended
An example is given in Appendix~\ref{section:1coveringalgorithm}.
\else 
An example is given in (CITE ARXIV) 
\fi




\subsection{Reed-Muller codes}
Reed-Muller codes are a central topic in coding theory, and defined as follows.

\begin{definition}
For a field~$\mathbb{F}_q$ and integers~$r\le (q-1)m$
an $RM_q(r,m)$ code is defined as the set of vectors 
\begin{align*}
    &RM_q(r, m) \triangleq\\ &\{(f(\alpha))_{\alpha\in\mathbb{F}_q^m} : f \in \mathbb{F}_q[x_1, x_2, \ldots , x_m], \deg(f) \leq r\}
\end{align*}
Furthermore, because $x^q = x$, we only consider polynomials where the degree of each $x_i$ is less than $q$. 
\end{definition}

The binary codes~$RM_2(r,m)$ can also be defined recursively using the so-called ``$(u,u+v)$ construction''
\begin{align*}
    &RM_2(r,m) \triangleq\\
    & \{(u,u+v): u \in RM_2(r,m-1), v \in RM_2(r-1,m-1) \}.
\end{align*}

Additionally, the GHWs for binary Reed-Muller codes are known. Let $\rho(r, m) = \sum_{i=0}^r {m\choose i}$ be the dimension of an $RM(r,m)$ code, and define the the \textit{canonical $(r,m)$-representation} ($(r,m)$-representation, for short) of a number~$t$ as follows: 

\begin{theorem} \label{decomposition}
\cite{hammingweights} Given $r,m$, any $0 \leq t \leq \rho(r,m)$ can be written as $$t = \sum_{i=1}^k \rho(r_i, m_i)$$ where the $r_i$ are decreasing, and $m_i - r_i = m - r  - i + 1$. In addition, this representation is unique. 
\end{theorem}

\begin{example}
    The canonical representation of~$7$ is $7 = \rho(1,4) + \rho(0,2) + \rho(0,1) $, since we have that~$\rho(1,4)=5$, $\rho(0,2)=1$, and $\rho(0,1)=1$.
\end{example}

Theorem~\ref{decomposition} is used to characterize the GHW hierarchy of Reed-Muller Codes. 

\begin{theorem} \label{Reed-Muller GHWs} 
\cite{hammingweights} $d_t(C) = \sum_{i=1}^k 2^{m_i}$, where $t = \sum_{i=1}^k \rho(r_i, m_i)$.

\end{theorem} 
Similarly, a slightly more involved expression for the GHWs of~$q$-ary Reed-Muller codes is known for~$q>2$~\cite{qaryRM}, and will be discussed in the sequel.


\section{The bound and the algorithm}

\subsection{A simple bound}
In this section we devise a bound on the generalized covering radius of a given code using its GHWs. The bound is based on Theorem~\ref{ghwcoveringbound} alongside Definition~\ref{def:generalizedradius4}, and Lemma~\ref{lemma:GHWscoincide} which follows, whose proof is trivial given the following alternative definition of~$d_r(C)$.
\begin{definition}\cite{hammingweights}
   $ d_r(C)= \min \{|I|:  I\subseteq [n], |I|-\operatorname{rank}(H_I)\geq r\} $, where $H$ is a parity check matrix and $H_I$ denotes the submatrix with columns of indices in I. 
\end{definition}

\begin{lemma}\label{lemma:GHWscoincide}
Let~$C$ be an~$[n,k]_q$ linear code with generator matrix~$G$, and for~$1\le t\le n-k$ let~$C_t\triangleq\{ xG\vert x\in\mathbb{F}_{q^t}^k  \} $. Then, the GHWs of~$C$ and~$C_t$ coincide.
\end{lemma}

\begin{proof} \label{coincidence}
Since $C$ and $C_t$ have the same parity check matrix, it follows trivially from the above definition that they also have the same GHWs. 
\end{proof}
An equally simple proof can be derived from Lemma 4 of \cite{KLOVE1992311}, however we have not explored if there are further connections to this work. 

We are now in a position to state the bound for the generalized covering radius. It is a straightforward combination of Lemma~\ref{lemma:GHWscoincide} with Theorem~\ref{ghwcoveringbound} and Definition~\ref{def:generalizedradius4}.
\begin{theorem} \label{theorem:bound}
$$R_t(C) \leq n - \sum_{r=1}^k \Bigl\lceil \frac{d_r - d_{r-1}}{q^t} \Bigl \rceil \triangleq \mu_t(C).$$
\end{theorem}

\begin{proof}
Since the GHWs  of $C$ and $C_t$ coincide, and the generator matrices of $C$ are also the generator matrices of $C_t$, it follows that $C_t$ satisfies the chain condition. Thus, one can apply Theorem~\ref{ghwcoveringbound} to $C_t$ and conclude the proof. 
\end{proof}

In the following subsections the bound from Theorem~\ref{theorem:bound} is used to obtain an efficient generalized covering algorithm for any chained code, and then specified to binary and nonbinary Reed-Muller codes.

\subsection{The algorithm.}
The following algorithm follows the outline of the one which follows from Theorem~\ref{ghwcoveringbound} for the ordinary covering problem; an example of which is given in \ifextended Appendix~\ref{section:1coveringalgorithm} \else (CITE ARXIV) \fi. It applies to any chained code, and returns a set of codewords which cover the input words up the value defined by the bound in Theorem~\ref{theorem:bound}. We assume that a chained matrix (Def.~\ref{matrixchaincondition}) is given as input; as an example, in the sequel it is shown how a chained matrix can be found for Reed Muller codes. In this algorithm we assume some fixed basis~$b_1,\ldots,b_t$ of~$\mathbb{F}_{q^t}$ over~$\mathbb{F}_q$, and denote by~$\binom{[n]}{\le\ell}$ the family of all subsets of~$[n]$ of size at most~$\ell$.



\begin{algorithm}
\caption{A $t$-covering algorithm for any chained code~$C$ with GHWs $d_1,\ldots,d_k$.}\label{alg:cap}
\begin{algorithmic}[1]
\State \textbf{Input:} A chained generator matrix $\Gamma$ for~$C$ and vectors $v_1,\ldots,v_t\in\mathbb{F}_q^n$ to cover.

\State \textbf{Output:} Codewords~$c_1,\ldots,c_t\in C$ such that for every~$i\in[t]$ we have~$\operatorname{supp}(v_i-c_i)\subseteq I$ for some~$I\in\binom{[n]}{\le \mu_t(C)}$. 
\State $u = u_0 \gets \sum_{i=1}^t b_iv_i$ (notice that~$u\in\mathbb{F}_{q^t}^n$).
\For{$i\in\{k,k-1,\ldots,1\}$} 
 \State $r_i \gets$ the~$i$-th row of $\Gamma$.
 \State Find~$a_i\in\mathbb{F}_{q^t}$ which maximizes $$\lvert\left\{ j\in \{ d_{i-1}+1,\ldots,d_i \}: u_j-a_ir_{i,j}=0\right\}\rvert.$$
 \State $u \gets u - a_i r_i$
\EndFor
\State Represent~$u - u_0=\sum_{i=1}^t b_i\ell_i$ for some~$\ell_i$'s in~$\mathbb{F}_q^n$.
\State \textbf{Return}  $\{\ell_i\}_{i=1}^t$.
\end{algorithmic}
\end{algorithm}

\begin{theorem}
The vectors~$\ell_1,\ldots,\ell_t$ returned by Algorithm~\ref{alg:cap} are codewords of~$C$, and there exists~$I\in \binom{[n]}{\le \mu_t(C)}$ such that~$\operatorname{supp}(v_i-\ell_i)\subseteq I$ for every~$i\in[t]$.
\end{theorem}

\begin{proof}
Observe that~$u-u_0=-\sum_{i=1}^k a_i r_i$ for some~$a_i$'s in~$\bF_{q^t}$, and hence there exist~$a_{i,j}$'s in~$\bF_q$ such that
\begin{align*}
    u-u_0=-\sum_{i=1}^k\left( \sum_{j=1}^t a_{i,j} b_j \right)r_i=\sum_{j=1}^tb_j\left( -\sum_{i=1}^k a_{i,j}r_i \right).
\end{align*}
Therefore, since the~$r_i$'s are rows in a generator matrix of~$C$, it follow that~$\ell_j=-\sum_{i=1}^k a_{i,j}r_i$ is a codeword of~$C$ for every~$j\in[t]$.

To show that a set~$I\in\binom{[n]}{\le \mu_t(C)}$ exists, i.e., that $\{\ell_j\}_{j=1}^t$ are a~$t$-covering of~$\{v_j\}_{j=1}^t$ within the bound, first observe that since~$u-u_0=-\sum_{i=1}^k a_i r_i$, it follows that~$u-u_0\in C_t$. Hence, we have found a codeword in~$C_t$ which covers~$u_0$, and since~$-\sum_{i=1}^k a_ir_i=\sum_{i=1}^t b_i\ell_i$ by definition, the respective Hamming distance is
\begin{align*}
    d_H(u_0,-\textstyle\sum_{i=1}^k a_ir_i)&=d_H(u_0,\textstyle\sum_{i=1}^tb_i\ell_i)\\
    &=w_H(\textstyle\sum_{i=1}^t b_iv_i - \textstyle\sum_{i=1}^t b_i \ell_i ) \\
    &= w_H (\textstyle\sum_{i=1}^t b_i(v_i - \ell_i)) \\
    &= \left|\textstyle\bigcup_{i=1}^t \operatorname{supp}(v_i - \ell_i)\right|\le \mu_t(C),
\end{align*}
where the last step follows by identical arguments as in the proof of Theorem~\ref{ghwcoveringbound}. Hence, the set~$I$ on which~$u_0$ and~$\sum_{i=1}^t b_i\ell_i$ differ belongs to $\binom{[n]}{\le \mu_t(C)}$, which concludes the proof.
\end{proof}

\subsection{Complexity analysis}



\begin{theorem} \label{runtime}
Algorithm~\ref{alg:cap} runs in $O(ntk\log(q)^2)$ time. 
\end{theorem}
\begin{proof}
To find the element~$a\in\mathbb{F}_{q^t}$ in the for-loop, recall that~$\{ r_{i,j} \}_{j\in\{d_{i-1}+1,\ldots,d_i\}}$ are nonzero for every~$i\in[k]$, and compute
$\{\{u_jr_{i,j}^{-1}\}\}_{j \in \{d_{i-1}+1,\ldots,d_i\}}$. Let~$a$ be the most frequently occurring value in this multi-set; it is readily verified that the required maximum is obtained by this value.

Computing $r_{i,j}^{-1}$ takes $O(\log(q)^2)$ time by the extended Euclidean algorithm, and multiplication of the $u_jr_{i,j}^{-1}$ takes $O(t\log(q)^2)$ time. This leads to complexity $O(t\log(q)^2\cdot(d_{i} - d_{i-1}))$ for each~$i$. Summation over all~$i$'s yields~$O(\log(q)^2 \cdot n)$ due to the telescopic sum.

Furthermore, computing~$u$ and $u - a\cdot r_i$ for each~$i$ takes $n(t\log(q)^2+t\log(q))$ time, and this computation is done~$k$ times, taking $O(n t k \log(q)^2)$ time in total. Thus, in total the algorithm runs in $O(ntk\log(q)^2)$ time. 
\end{proof}


\section{Application to Reed Muller codes}

In this section we specialize our techniques for Reed-Muller codes due to their importance and their interesting covering properties~\cite{elimelech2022rmgeneralized}. The challenge in applying our techniques for any code~$C$ is finding the generator matrix~$\Gamma(C)$. This allows us to apply our algorithm to any Reed-Muller code, as opposed to just binary ones. First, we show how to find a chained matrix for $q$-ary RM codes, which allows us to use our algorithm. Then, for binary RM codes, we devise an extension of~\cite{elimelech2022rmgeneralized} using our algorithm, which improves upon~\cite{elimelech2022rmgeneralized} exponentially.
    \subsection{Chained generator matrices for Reed-Muller codes}\label{section:chainedMatrices}
    For~$RM_2(r,m)$, there exists a chained generator matrix whose~$j$'th row is the evaluation of the the~$j$'th monomial of degree~$r$ or less, according to  \textit{anti-lexicographic} order; the full details are given in~\cite[Thm.~7]{hammingweights}. For the general~$q$-ary case, we provide the following generalization of the~$(u,u+v)$ construction. 
    
    \begin{theorem}\label{theorem:RMgeneratormatrixq}
    Let $RM_q(r,m)$ be a $q$-ary RM code with generator matrix $G_{r,m}$. Then letting $w =\min(r,q-1)$ and $\bF_q=\{0,\gamma^0,\gamma^1,\ldots,\gamma^{q-2}\}$, we have that $G_{r,m} =$ 
    \begin{align*}
        \begin{bmatrix} (\gamma^{q-2})^wG(r-w,m-1) & ... & {0}^wG(r-w,m-1) \\ \vdots & \vdots & \vdots \\ (\gamma^{q-2})^{0}G(r,m-1) & ... & 0^0G(r,m-1)\end{bmatrix}.
    \end{align*}
    \end{theorem}
    
    Intuitively, Theorem~\ref{theorem:RMgeneratormatrixq} holds by considering the multivariate polynomial ring in which the Reed-Muller code exists as a univariate polynomial ring over the last variable $x_m$. Then, using a standard basis consisting of all monomials of total degree~$r$ or less, order the rows based on the degree of $x_m$, and order the columns based on the field element that~$x_m$ evaluates to.
    

    When viewing the generator matrix in this block form, we see that because $G(r-i,m-1)$ is a subcode of $G(r,m-1)$, we can ``subtract'' lower block rows (in the block form) from upper block rows (this corresponds to subtracting matching monomials from each other, after evaluating them at $x_m$). However, we cannot use upper block rows to row-reduce lower block rows. This leads to the following theorem. 

    \begin{theorem}\label{theorem:rowReduction}
    There exists a row reduction (in block form) after which $G_{r,m}$ becomes lower triangular in the block form presented in Theorem~\ref{theorem:RMgeneratormatrixq}. 
    \end{theorem}

    \ifextended
    \begin{proof}
        We present such row-reduction as follows, where the block-columns are numbered from~$0$ (rightmost) to~$q-1$ (leftmost). The $0$'th block-column is already in lower-triangular form, since~$0^0G(r,m-1)$ (the bottom block) is the only nonzero block in that block-column. Next, in the~$1$'st block-column we use multiples of~$(\gamma^0)^1 G(r-1,m-1)$ to zero-out all the blocks $(\gamma^0)^2 G(r-2,m-1),\ldots,(\gamma^0)^w G(r-w,m-1)$ above it; this is possible due to the subcode property mentioned above. In general, in block-column~$i\ge 1$ we use multiples of $(\gamma^{i-1})^{i} G(r-i,m-1)$ to zero-out all the blocks $(\gamma^{i-1})^{i+1} G(r-i-1,m-1),\ldots,(\gamma^{i-1})^w G(r-w,m-1)$ above it; this does not spoil the zero blocks to the right of block-column~$i$ since they are already zero.
    \end{proof}

    \else 
    \begin{proof}
    Given in (CITE ARXIV)
    \end{proof}
    \fi
    
    Given this lower triangular form, we have that a chained matrix for~$RM_q(r,m)$ can be obtained recursively from RM codes of lower order. While seeing that it is a generator matrix of~$RM_q(r,m)$ is rather straightforward, showing that it indeed realizes the GHWs of~$RM_q(r,m)$ requires a dive into their definition from~\cite{qaryRM}.

    \begin{theorem}\label{theorem:RMqChainedMatrix}
    Consider an $RM_q(r,m)$ with generator matrix $G_{r,m}$. Then 
    $$\Gamma(G_{r,m}) =
    \begin{bmatrix} 
    \Gamma(G(r-w,m-1)) & ... & 0\\
    \vdots & \ddots & \vdots \\
    \Gamma(G(r,m-1)) & ... & \Gamma(G(r,m-1))
    \end{bmatrix}.$$
    \end{theorem}

    \ifextended
    \begin{proof}

    
     We can see that this is a generator matrix directly from Theorem~\ref{theorem:RMgeneratormatrixq}, and will prove that it has the correct form using induction. Our base cases are when $r = 0$ or when $r = (q-1)m$. When $r = 0$, the generator matrix is just a row of all $1$'s, and when $r = (q-1)m$, then we have the identity matrix, both of which are chained. In the inductive case, it is obvious that if the matrix is lower triangular, and the matrices on the diagonal blocks are chained (for the respective codes spanned by them), then the overall matrix is of the correct form. 

     It remains to be seen, however, if the rows of this matrix realize the correct generalized hamming weights. In \cite{qaryRM}, the generalized Hamming weights of $q$-ary Reed-Muller codes are established. Furthermore, they give a recursive algorithm which can be used to compute the GHWs. We see that if our matrix in fact realizes these weights, then it would also induce a recursive way of computing the GHWs (by computing the GHWs of each of the diagonal blocks). In fact, doing so follows the exact same recursive algorithm that is given in \cite[Remark~6.9]{qaryRM}. We will go through the four parts of the algorithm to show this. The algorithm runs on a $4$-tuple $(r,u,v,m)$. We will view the algorithm as finding the $r$-th row of $\Gamma(G_{r,m})$ by recursively adding up the lengths of the blocks. In this case $u,m$ are the code parameters of the current block, $v$ is the block row of the matrix (also the degree of the variable we are considering), and $r$ is the row whose length we are trying to find. The four steps enumerated in the paper are the following: 

    i) $v = 0$: In this case, we have made it to the last row of the matrix, so we must move to the next variable and make a recursive call on $u,m-1$. 

    ii) $v>u$: degree of variable we are considering is higher than max degree of polynomials in the code, therefore since there are no variables of degree higher than $u$ we set $v := u$. 

    iii) $r > \rho(u-v,m-1)$: then the number of rows in the block row we are considering is greater than $r$, so we can compute the length of $G_{u-v,m-1}$ (the element on the diagonal of the row we are considering), and move on to the next row, adjusting $r$ by the number of rows we just computed. The tuple that is offshooted in this step is simply the length of $G_{u-v,m-1}$. At the end of the algorithm all of these lengths are added up, equivalent to adding up the lengths of the blocks on the block diagonal of $\Gamma(G_{r,m})$.

    iiii) $r <\rho(u-v,m-1)$: Then the row we need is within the block row we are consdering, so we must recurse on the block element within the row we are considering to get to the correct row. 

    We see that these four steps, equivalent to the ones described in the paper, describe finding the support of the first $r$ rows of $\Gamma(G_{u,m})$. Therefore $\Gamma(G_{u,m} )$ realizes the GHW hierarchy of $RM_q(r,m)$. 
    \end{proof}

    \else
    \begin{proof}
     Given in (CITE ARXIV)
     \end{proof}
    \fi 
    \subsection{A modification to Elimelech's covering algorithm}
    Elimelech's algorithm only applies to the binary $RM_2(r,m)$, and relies on its recursive~$(u,u+v)$ construction. Roughly speaking, it receives as input a~$t\times n$ matrix of~$t$ vectors to cover, splits it in half lengthwise, and calls a routine named RECURSIVE to cover each of the halves individually. Additionally, it attempts to employ the subadditive property mentioned in Theorem~\ref{subadditive} using a routine called SUBADDITIVE, which calls RECURSIVE with each row of the input matrix, and returns the minimum of the two. 
    In the base case it requires a brute force computation of the minimum distance between all codewords of $RM_2(1,m)$ and a fixed vector~$v$, which leads the complexity to be quadratic in the length~$n=2^m$ and exponential in~$t$.
    
    While Algorithm~\ref{alg:cap} can be applied directly on any~$RM_2(r,m)$, a better alternative exists. Specifically, we follow the recursive structure of Elimelech's algorithm and only replace the base case by Algorithm~\ref{alg:cap} (for which a chained matrix is easily computable). This modification reduces the complexity to be linear in both~$n$ and~$t$ but comes at a cost---the output codewords will cover the input vectors only up to the bound in Theorem~\ref{theorem:bound}. 
    Hence, in \ifextended Appendix~\ref{section:experimental} \else (CITE ARXIV) \fi we numerically compare known bounds from~\cite{elimelech2022rmgeneralized} against Theorem~\ref{theorem:bound}, and show improvement in many cases.
    Let COVER be the algorithm from~\cite[Alg.~1]{elimelech2022rmgeneralized}, for which we have the following.
    \begin{theorem}\label{ElimelechRuntime}
    \cite[Thm.~25]{elimelech2022rmgeneralized} For any $t,r,m \in \mathbb{N}$, $\mbox{COVER}(v,r)$ has complexity $$O(t2^t2^{(t+1)(\log(n)+1)}(2^{t+1}-1)^{-r} + tn\log(n))).$$
    \end{theorem}
    
    Let~$\mbox{COVER}'$ be the version of COVER with the modified base case, for which we have the following.
    
    \begin{theorem} 
    For any $t,r,m \in \mathbb{N}$, $\mbox{COVER}'(v,r)$ has complexity $O(tn \log(n))$.
    \end{theorem} 
    
    \begin{proof} 
    This proof follows the same methodology as the proof of Theorem~\ref{ElimelechRuntime}. We analyze the complexity of $\mbox{RECURSIVE}(v,r)$, denoted by~$T(t,r,m)$, in an inductive manner.
    
    We have two base cases. When $r=m$ we have that $T(t,m,m) = c'$ for some constant~$c'$. When $r=1$ we apply Algorithm~\ref{alg:cap} in~$O(nkt)$ time with $k = \log(2^m) + 1 = \log(n) + 1$. Hence, both base cases run in~$O(tn\log(n))$ time.
    
    In our inductive step, we assume that the claim holds for~$T(t,r-1,m-1)$ and~$T(t,r,m-1)$, and prove that it holds for $T(t,r,m)$. 
    The algorithm splits a matrix of size $t\times n$ in half and then calls two recursive instances. Thus, for some constants~$c$ and~$c'$ we have that
    \begin{align*}
    T(t,r,m) &= c' tn + T(t,r-1,m-1) + T(t,r,m-1)  \\
    &\le c' tn + 2ct(n/2)\log(n/2)=O(tn\log(n)).
    \end{align*}
    This completes the proof of $\mbox{RECURSIVE}(v,r)$. To complete the proof of the overall algorithm, notice that $\mbox{SUBADDITIVE}$ is merely~$t$ consecutive calls to $\mbox{RECURSIVE}$ with $t = 1$, and hence it does not increase the overall complexity. 
    \end{proof} 
\newpage
\printbibliography
\clearpage

\ifextended
\appendices

\section{A covering algorithm}\label{section:1coveringalgorithm}
We demonstrate the algorithm which follows from the proof of Theorem~\ref{ghwcoveringbound}; it is also a special case (with~$t=1$) of Algorithm~\ref{alg:cap}. We choose the code~$C=RM(1,3)$, which is a chained code whose GHWs are~$(4,6,7,8)$ (Theorem~\ref{Reed-Muller GHWs}), and the value of the bound for it is~$3$.


Consider the following chained matrix (see Section~\ref{section:chainedMatrices}),
\begin{align*}
    \Gamma(C) =  \begin{bmatrix} 1 & 1 & 1 & 1 & 0 & 0 & 0 & 0 \\ 1 & 1 & 0 & 0 &1 & 1 & 0 &0 \\ 1 & 0 & 1 & 0 & 1 & 0 & 1 & 0 \\ 1 & 1 & 1 & 1 & 1 & 1 & 1 & 1  \end{bmatrix}  \triangleq
    \begin{bmatrix}
    g_1\\g_2\\g_3\\g_4
    \end{bmatrix}
\end{align*}
and suppose we wish to find the closest codeword to $v_0 = (v_1,\ldots,v_8)=(1,0,0,1,1,1,0,1)$. 
\begin{itemize}
\item Set~$v=v_0$.
\item \textit{Covering entry~$8$ using~$g_4$}: Since entry~$8$ of~$v$ and~$g_4$ coincide, we set~$v\leftarrow v - g_4 = (0,1,1,0,0,0,1,0)$.
\item \textit{Covering entry~$7$ using~$g_3$}: Since entry~$7$ of (the modified)~$v$ and~$g_3$ coincide, we set~$v\leftarrow v - g_3 = (1,1,0,0,1,0,0,0)$.
\item \textit{Covering entries~$5$ and~$6$ using~$g_2$}: Since $1\cdot(1,1)$ and $0\cdot(1,1)$ are equidistant from $(v_5,v_6) = (0,1)$, both substitutions $v\leftarrow v -0\cdot g_2$ and $v\leftarrow v -1\cdot g_2$ are eligible. Hence,
\begin{itemize}
    \item Choosing $v \leftarrow v-1\cdot g_2 = (0,0,0,0,0,1,0,0)$ concludes the algorithm, since $v_0 - v = (1,0,0,1,1,0,0,1)$ is a codeword close enough to~$v_0$ to be within our bound. 
    \item Choosing $v \leftarrow v-0\cdot g_2 = (0,0,0,0,0,1,0,0)$ requires proceeding to the next step.
\end{itemize}
\item \textit{Covering entries~$1$ through~$4$ using~$g_1$}: We again find that $0\cdot(1,1,1,1)$ and $1\cdot (1,1,1,1)$ are equidistant from $(v_1,\ldots,v_4)=(1,1,0,0)$, and hence either one can be selected. Choosing the coefficient~$1$, we get that $v = (0,0,1,1,1,0,0,0)$, so the codeword we find in this instance is $v_0 - v = (1,0,1,0,0,1,0,1)$. It is of distance~$3$ from $v$, hence within the bound as well.
\end{itemize}

\section{Experimental results}\label{section:experimental}
    Since the bound of Theorem \ref{ghwcoveringbound} does not have a closed form, we are left to compare it with the bounds given in~\cite{elimelech2022rmgeneralized} computationally. Asymptotically it appears that in most cases our bound is either equal or better than the others. Some example plots are given below, where~$M_t(r,m)$ is the best bound in~\cite{elimelech2022rmgeneralized} for the~$t$'th covering radius of~$RM_2(r,m)$, and $\mu_t$ is our bound. We note that this comparison is slightly unfair; bounds tighter than $M_t$ can be computed using their algorithm which does not have any closed form. 
    \begin{figure}[h]
  \centering
  \includegraphics[width=.8\linewidth]{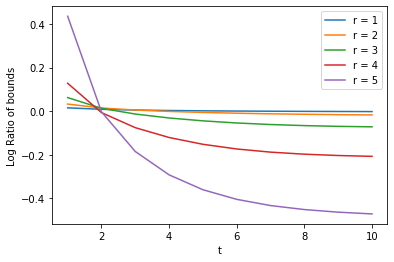}
  \caption[Comparison of upper bounds]{A plot of $\log(\mu_t(RM(r,12))/R_t(r,12))$ for various values of $r$ and $t$.} 
  \label{fig:sub1}
    \end{figure}%

    \begin{figure}[h]
  \centering
  \includegraphics[width=.8\linewidth]{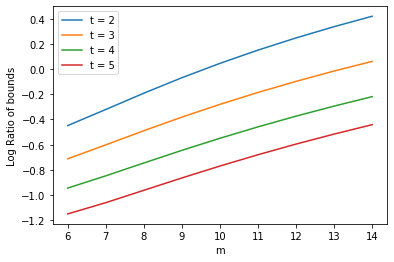}
  \caption[Comparison of upper bounds in $m$ for $s = 5$]{A plot of $\log(\mu_t(RM(m-5,m))/M_t(m-5,m))$ for various values of $m$ and $t$.}
  \label{fig:sub1}
\end{figure}%
\begin{figure}[h]
  \centering
  \includegraphics[width=.8\linewidth]{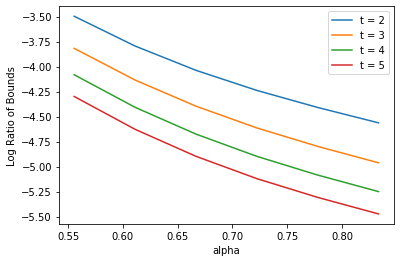}
  \caption[Comparison of upper bounds in $\alpha$ for $m = 18$]{A plot of $\log(\mu_t(RM(\alpha\cdot 18,18))/M_t(\alpha\cdot 18,18))$ for various values of $\alpha$ and $t$.}
  \label{fig:sub1}
\end{figure}
\begin{figure}[h]
  \centering
  \includegraphics[width=.8\linewidth]{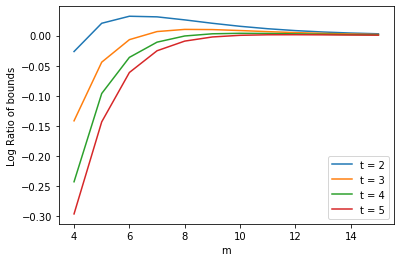}
  \caption[Comparison of upper bounds in $m$ for $r = 1$]{A plot of $\log(\mu_t(RM(1,m))/M_t(1,m))$ for various values of $m$ and $t$.}
  \label{fig:sub1}
\end{figure}%
\begin{figure}[h]
  \centering
  \includegraphics[width=.8\linewidth]{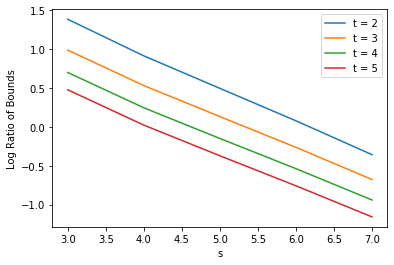}
  \caption[Comparison of upper bounds in $s$ for $m = 15$]{A plot of $\log(\mu_t(RM(15-s,15))/M_t(15-s,15))$ for various values of $s$ and $t$.}
  \label{fig:sub1}
\end{figure}%

\begin{figure}[h]
  \centering
  \includegraphics[width=.8\linewidth]{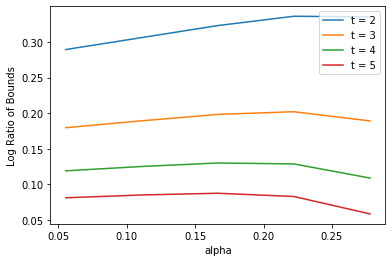}
  \caption[Comparison of upper bounds in $\alpha$ for $m = 18$]{A plot of $\log(\mu_t(RM(\alpha \cdot  18,18))/M_t(\alpha\cdot 18,18))$ for various values of $\alpha$ and $t$.}
  \label{fig:sub1}
\end{figure}%

\begin{figure}[h]
  \centering
  \includegraphics[width=.8\linewidth]{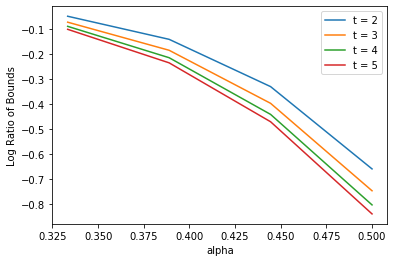}
  \caption[Comparison of upper bounds in $\alpha$ for $m = 18$]{A plot of $\log(\mu_t(RM(\alpha\cdot 18,18))/M_t(\alpha\cdot 18,18))$ for various values of $\alpha$ and $t$.}
  \label{fig:sub1}
  \end{figure}%
\begin{figure}[h]
  \centering
  \includegraphics[width=.8\linewidth]{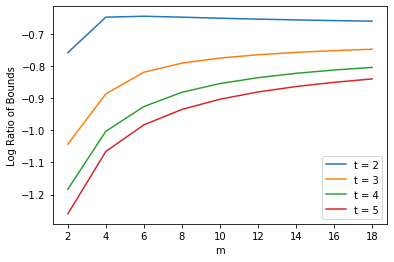}
  \caption[Comparison of upper bounds in $m$ for $\alpha = .5$]{A plot of $\log(\mu_t(RM(.5m,m))/M_t(.5m,m))$ for various values of $m$ and $t$.}
  \label{fig:sub1}
\end{figure}%
\begin{figure}[h]
  \centering
  \includegraphics[width=.8\linewidth]{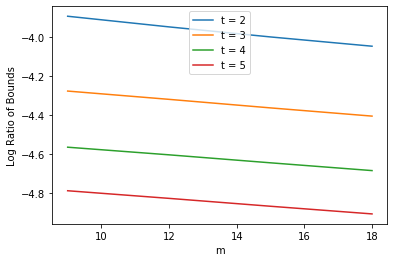}
  \caption[Comparison of upper bounds in $m$ for $\alpha = \frac{2}{3}$]{A plot of $\log(\mu_t(RM(\frac{2}{3}m,m))/M_t(\frac{2}{3}m,m))$ for various values of $m$ and $t$.}
  \label{fig:sub1}
\end{figure}%
\else
\fi
\end{document}